\documentclass[11pt]{article}
\usepackage{comment}
\usepackage[english]{babel}
\usepackage[latin1]{inputenc} 
\usepackage{graphicx}
\usepackage[table,xcdraw]{xcolor}
\usepackage[centertags]{amsmath}
\usepackage{amsfonts}
\usepackage{amssymb}
\usepackage{amsthm}
\usepackage{bm}
\usepackage{cite}
\usepackage{multirow}
\usepackage{geometry}
\usepackage{indentfirst}
\usepackage{dsfont}
\usepackage{textcomp}
\date{}
\usepackage{caption}

\begin{document}

\title{A taxonomy for controlling (in)consistency}

\author{
	{\large  \textsc{Marcelo E. Coniglio}}\thanks{coniglio@unicamp.br}\\
	{\small UNICAMP, CLE/IFCH} \\
	{\large \textsc{Rafael Ongaratto}}\thanks{ongarattorafa@gmail.com} \\
	{\small UNICAMP, CLE/IFCH} 
}

\maketitle



\begin{abstract}
  \noindent  In this article, the hierarchy of LFIs L$_n^k$, Logics of Controlled Consistency (LCC), is introduced. Inspired by da Costa's original C$_n$ systems, this hierarchy can represent different degrees of paraconsistent commitment and different related notions of consistency, inconsistency, and negation associated with each two-dimensional level of these logics. In one dimension, the logics become increasingly more paraconsistent by allowing the consistency operator to behave inconsistently up to a fixed iteration. In another dimension, the negation is increasingly strengthened. Initially, we present these logics with a swap structure semantics, showing their soundness and completeness. Some well-known LFIs are shown to be particular cases of LCCs. With some examples, we show how these different logics represent different types of paraconsistent commitment: from skepticism to dogmastism, these logics have the multiplicity to represent these different philosophical positions. Furthermore, the development of the hierarchy in a general manner allows pragmatism to take place when considering the different types of paraconsistent commitment. Each level we go up in this direction we get a stronger family of logics. Furthermore, we also present an extension of an LCC, a 5-valued LFI called LFI3, a sublogic of LFI1. LFI3 presents a paradigmatic case for the development of many-valued LFIs that have more than three values. Using a technique that combines Karnaugh Maps and Twist Structures, we give an axiomatization and a semantical account of LFI3. Finally, using RNmatrices, we give a general semantical account of the L$_n^k$ family of logics, and we also prove its soundness and completeness.

\end{abstract}


\section{Introduction}
Logics of Formal Inconsistency (LFIs) are a family of paraconsistent logics that includes a unary consistency operator $\circ$ that recovers classicality \cite{car:01}. By introducing such a primitive operator, LFIs are able to distinguish many scenarios in which consistency plays a role. We can consider stronger contexts in which non-contradiction implies consistency and consistency implies the consistency of consistency, such as the case of the robust 3-valued logic LFI1. However, we can also weaken the notion of consistency in such a way that consistency itself is sometimes contradictory: in mbC, the minimal LFI based on classical logic, consistency implies non-contradiction, which is guaranteed by the axiom (bc1) $\circ\alpha$ $\rightarrow$ ($\alpha$ $\rightarrow$ ($\neg\alpha$ $\rightarrow$ $\beta$)). However, neither non-contradiction implies consistency nor consistency implies consistency of consistency in mbC \cite[p. 33]{carnielli2016}.  

Furthermore, LFIs are also able to formalize inconsistency via the inconsistency operator $\bullet$. Whether to take $\circ$ or $\bullet$ as primitive operators is a matter of logical choice, but once the decision is taken, they are interdefinable. In that manner, we can also distinguish different inconsistent scenarios: in mbC, inconsistency does not imply contradiction. However, in LFI1, inconsistency and contradiction are both implied by each other. 

In intermediate LFIs, more axioms are introduced to restrict the behavior of $\circ$: mbCciw is the logic resulting from mbC with the addition of the axiom (ciw) $\circ\alpha$ $\lor$ ($\alpha$ $\land$ $\neg\alpha$). Adding (ciw) makes consistency determined in terms of non-contradiction. However, to determine inconsistency in terms of contradiction, we need to add the axiom (ci) $\neg$$\circ\alpha$ $\rightarrow$ ($\alpha$ $\land$ $\neg\alpha$), obtaining the stronger system mbCci \cite[p. 66-67]{carnielli2016}. 

The logics mbCciw and mbCci share a close relationship: mbCci is equivalent to mbCciw plus the axiom (cc) $\circ$$\circ\alpha$. Thus, they are an important step in the standard construction of LFIs, because the jump from mbCciw to mbCci is an important change in the behavior of consistency: in mbCci, any judgment of consistency is required to be itself consistent, characterized by the axiom (cc).

The jump from mbCciw to mbCci regarding consistency is not small: any formula of the form $\circ$$\circ\alpha$ is valid in mbCci. In particular, take a formula $\alpha$ = $\circ^n$$\beta$, $n$ $\geq$ 0. By (cc), $\vdash_{mbCci}$ $\circ^{n+2}$$\beta$. By a straightforward procedure, we can show that, for any formula $\alpha$, and for any k $\geq$ 2, $\vdash_{mbCci}$ $\circ^k$$\alpha$. 

This means that consistency in mbCci is only allowed to be inconsistent regarding judgments that are not about the consistency of some other proposition. Any level of meta-consistency, that is, consistency about the consistency of a statement, is not allowed to include divergences. The situation in mbCciw is the opposite: consistency is possibly contradictory at any level. There is no iteration $l$ of $\circ$ such that for all formulas $\alpha$, $\vdash_{mbCciw}$ $\circ^l$$\alpha$. Therefore, either you can choose consistency to behave inconsistently all the way in mbCciw, or you can choose consistency to be consistent in all meta-consistent situations in mbCci. This appears to be contrary to the spirit of paraconsistency, which already represents the notion of \textit{degrees} of paraconsistent commitment in da Costa's original work on C$_n$ systems \cite{da1974theory}. Even though C$_n$ systems were proven to be LFIs \cite[p. 25]{marcos2007}, the consistency operator in C$_n$ systems is a definable operator, that is, they are also dC-systems \cite[p. 72]{carnielli2016}. Thus, LFIs still lack a proper hierarchy of logics of degrees of paraconsistent commitment that is comparable to C$_n$ systems.

The above comparison already suggests that the construction of LFIs in this point should be supplemented. In many situations, meta-consistent statements are required to be not always consistent. For example, when a lawyer talks about the testimony of someone who he regards as unreliable, the prosecutor can talk about the lawyer's statements as themselves unreliable. That is, the prosecutor considers the lawyer's statements about the consistency of a testimony as themselves inconsistent.  In general, the context also dictates to which extent we can question the consistency of judgments. For example, when playing a game of detective in which the players can lie only about their current location, the players can question the consistency of each other's statement about their location. However, it makes no sense within this game to question the consistency of the statements concerning the consistency of their location, because they can only lie about their location. In this situation, consistency would be controlled at two iterations of the consistency operator. There may be other situations in which we want to control consistency at four iterations. In general, we want to control consistency statements to the point that we can control consistency at any level $n$, $n$ $\in$ $\mathds{N}$.  For each level of control $n$, we will present a family of LFIs to characterize such a context. mbCciw and mbCci stand as extreme points of this hierarchy of logics. Starting from mbCciw, for each logic L$_n^k$, we add the axiom (cc$_n$) $\circ^{n+2}$$\alpha$.  

Nevertheless, we can also go beyond this comparison. It is possible to strengthen mbCciw in other directions, while preserving the need for a (cc$_n$) axiom. One level up, we consider the system mbCciw with the addition of the axioms (cf) $\neg\neg$$\alpha$ $\rightarrow$ $\alpha$ and (ce) $\alpha$ $\rightarrow$ $\neg\neg\alpha$, resulting in the systems L$_n^1$. In level L$_n^2$, we consider the system L$_n^1$ plus the axiom (ip$^2$) $\neg{\circ}$$\neg\alpha$ $\leftrightarrow$ $\neg{\circ}$$\alpha$. Finally, in L$_k^3$, we consider the system L$_k^2$ plus the axiom (ip$^3$) $\neg{\circ}$$\circ$$\neg\alpha$ $\leftrightarrow$ $\neg{\circ}$$\circ$$\alpha$. Each level we go up in this direction we get a stronger family of logics.

Finally, we present a semantics based on swap structures \cite{coniglio2020} for some basic LCCs. Moreover, by inducing Nmatrices, we will show how to present these logics in terms of many-valuedness, providing an intuitive reading of these logics. From another perspective, a uniform semantics based on RNmatrices is presented for the logics L$_k^n$. Along the way, we will meet some already known logics such as Cbr, Cie \cite{coniglio2025logicparaconsistentbeliefrevision}, and we present a new 5-valued logic, LFI3, which has interesting relations with LFI1, an LFI whose history is deeply entangled with the history of paraconsistency: originally presented in \cite{itala1970} under the name of J3 and a different signature, LFI1 represents a starting point in the history of LFIs (see~\cite{carnielli2000} and~\cite[Section~4.4]{carnielli2016}).

\newtheorem{corollary}{Corollary}
\newtheorem{definition}{Definition}

\section*{A First Characterization of L$_n^k$}

In the most basic case of LFIs, consistency simply implies explosion.
\begin{center}
    (bc1) $\circ\alpha \rightarrow (\alpha \rightarrow (\neg\alpha \rightarrow \beta))$
\end{center}
The basic LFI based on classical logic is mbC, which is composed of Classical Positive Logic (CPL$^+$) and the axiom (TND) $\alpha$ $\lor$ $\neg\alpha$, together with the axiom (bc1).
\begin{definition}\label{mbc}
    mbC := CPL$^+$ + (TND) + (bc1) 
\end{definition}
In the standard construction of LFIs, some axioms are added to restrict the behavior of $\circ$ or of $\neg$. The first step towards this construction reaches mbCciw, which adds the following axiom to mbC:
\begin{center}
    (ciw) $\circ\alpha \lor (\alpha \land \neg\alpha)$
\end{center}
\begin{definition}\label{mbcciw} mbCciw := mbC + (ciw)
\end{definition}

As stated before, mbCciw only determines consistency in terms of non-contradiction. It does not guarantee that consistency (at any level) is consistent, that is, $\nvdash_{mbCciw}$ $\circ^n$$\alpha$, n $>$ 0. In order to avoid this feature, mbCci adds the axiom (cc) $\circ\circ$$\alpha$, thus guaranteeing that any meta-consistent statement is itself consistent. 

\begin{definition}\label{mbcci}
    mbCci := mbCciw + (cc) $\circ$$\circ$$\alpha$
\end{definition} 

\newtheorem{theorem}{Theorem}

\begin{theorem}\label{bolinha dupla teorema}
    Let k $\geq$ 2. Then, $\vdash_{mbCci}$ $\circ^k$$\alpha$.
\end{theorem}

In the standard construction of LFIs, there is a significant step made between mbCci and mbCciw, namely, there is a jump from absolute freedom of inconsistency to restriction of inconsistency only to first-level sentences. Similarly to Costa's C$_n$ systems and its hierarchy of increasingly stronger paraconsistent logics, we introduce the hierarchy L$_n^k$ of LFIs of Controlled Consistency:

  \begin{definition}\label{definicao geral}
  (1) Let L$_n^0$ := mbCciw + (cc$^n$) $\circ^{n+2}$$\alpha$, where $n \geq 0$.\\
  (2) Let
        L$_n^1$ := L$_n^0$ + (dn) $\neg\neg\alpha$ $\leftrightarrow$ $\alpha$. \\
  (3) Let
        L$_n^k$ := L$_n^0$ + (dn) $\neg\neg\alpha$ $\leftrightarrow$ $\alpha$  + (ip$^j$) $\neg$$\circ$$^{j}$$\neg\alpha$ $\leftrightarrow$ $\neg{\circ}^{j}$$\alpha$, where $n \geq 0$ and $1 \leq j < k$.
    \end{definition}
    
This hierarchy is two-dimensional: given L$_n^k$, the dimension $n$ tells us the level at which consistency becomes itself consistent, because it determines (cc$^n$) $\circ^{n+2}$$\alpha$. On the other hand, the dimension $k$ tells us the strength of the negation: as the level $k$ grows, the negation behaves more deterministically.

\section*{L$_1^k$: Basic Logics of Controlled Consistency}

Consider the language $\mathcal{L}_\Sigma$ defined by a set of propositional atoms $P = \{p_1,p_2,p_3, \ldots\}$ (using $p, q, r$ as metavariables) and the following BNF definition: 
\[
\alpha := p\ |\ \neg\alpha\ |\ \circ\alpha\ |\ (\alpha \land \alpha)\ |\ (\alpha \rightarrow \alpha)\ |\ (\alpha \lor \alpha)\  
\]
The logics LCC are defined over the language $\mathcal{L}_{\Sigma}$, $\Sigma$ is the signature. To start, let us consider L$_1^0$.

\begin{definition}\label{l_um_zero}
    L$_1^0$ := mbCciw + (cc$^1$) $\circ$$\circ$$\circ$$\alpha$
\end{definition}
L$_1^0$ is called the basic Logic of Controlled Consistency. Contrary to mbCci, this logic does not validate (cc) $\circ$$\circ$$\alpha$, that is, it allows inconsistency one level up from mbCci. From (cc$^1$), it follows that $\circ$$\circ$$\alpha$, $\neg{\circ}$$\circ$$\alpha$ $\vdash$ $\beta$. However, $\circ$$\alpha$, $\neg{\circ}$$\alpha$ $\nvdash$ $\beta$. That is, contradictoriness about the consistency of the consistency of some statement implies explosion, but not contradictoriness about the consistency of some statement. 
In the following definitions, a swap structure semantics will be constructed for L$_1^0$ and its respective Nmatrices induced. 

It will be useful in the sequel consider the two-element Boolean algebra $\mathcal{A}_2$ with domain ${\bf 2}=\{0,1\}$ and where meet, join, Boolean implication and Boolean complement will be denoted (without risk of confusion from the context) by $\land$, $\vee$, $\to$ and $\sim$, respectively.

Recall that swap structures are non-deterministic matrices whose elements are $(n+1)$-tuples (called snapshots) representing a 0-1 value for the formulas $\alpha$, $\psi_1(\alpha)$, \ldots, $\psi_n(\alpha)$ for given formulas  $\psi_1(p)$, \ldots, $\psi_n(p)$ depending on a single variable $p$. In turn, its multioperations are defined for each coordinate as a suitable Boolean combination of the input coordinates (see a detailed description of swap structures in~\cite{coniglio2025}).

Snapshots for L$_1^0$ will be triples $x=(x_1,x_2,x_3)$ representing 0-1 values for $\alpha, \neg\alpha$ and $\neg{\circ}\alpha$. By (TND) $x_1 \vee x_2=1$.
Since $\circ\alpha$ is equivalent to ${\sim}(\alpha \land \neg \alpha)$ (where ${\sim}\beta:=\neg \beta \land {\circ}\beta$ is the classical negation definable in mbCciw), it is not necessary to keep control of $\circ\alpha$ in the snapshots: $(\circ x)_1={\sim}(x_1 \land x_2)$ for any snapshot $x$. This implies, by (TND), that $x_3 \vee {\sim}(x_1 \land x_2)=1$. Using this  and the fact that $\circ$$\circ$$\circ$$\alpha$ is valid, the formula ${\circ}{\circ}\alpha \land \neg{\circ}{\circ}\alpha$ is a bottom. By (TND), it follows that $\neg{\circ}{\circ}\alpha$ is equivalent to ${\sim}{\circ}{\circ}\alpha$, which in turn is equivalent to $\circ \alpha \land \neg{\circ}\alpha$. That is, $(\circ x)_3={\sim}(x_1 \land x_2) \land x_3$, because of the snapshots definition. Since there are no restrictions on the values of $\neg(\alpha \#\beta)$ and $\neg{\circ}(\alpha \#\beta)$ for $\# \in \{\land,\vee,\to\}$, we arrive at the following definition:

\begin{definition}\label{semantica_1}
    Let $\mathds{B}_1^0 = \{x \in {\bf 2}\times {\bf 2}\times {\bf 2} \ : \ x_1 \lor x_2 = 1 \mbox{ and } x_3 \lor {\sim}(x_1 \land x_2) = 1\}$. The swap structure for L$_1^0$  is the multialgebra 
        $\mathcal{B}_1^0$ = $\langle \mathds{B}_1^0$, $\land$, $\lor$, $\rightarrow$, $\neg$, $\circ$$\rangle$  over $\Sigma$
such that the multioperations are defined as follows, for every $x$ and $y$ in $\mathds{B}_1^0$:

$\begin{array}{lll}
   x \# y  & = & \{z \in \mathds{B}_1^0 \ : \ z_1 = x_1 \# y_1\}, \text{ for each }\# \in \{\land, \lor, \rightarrow\};\\
   \neg x  & = & \{z \in \mathds{B}_1^0 \ : \ z_1 = x_2\};\\
  \circ x  & = & \{({\sim}(x_1 \land x_2), x_3, x_3 \land {\sim}(x_1 \land x_2))\}.
\end{array}$

\noindent Let $D_1^0 = \{x \in \mathds{B}_1^0 \ : \ x_1 = 1\}$ be the set of designated values. The non-deterministic matrix associated to $\mathcal{B}_1^0$ is $\mathcal{M}_1^0 = (\mathcal{B}_1^0, D_1^0)$, where $\vDash$$_{\mathcal{M}_1^0}$ is the corresponding consequence relation.
\end{definition}

Observe that $\mathds{B}_1^0 = \{(1, 0, 0), (1, 0, 1), (1, 1, 1), (0, 1, 1), (0, 1, 0)\}$, while the set of designated values is $D_1^0= \{(1, 0, 0), (1, 0, 1), (1, 1, 1)\}$.
Labeling $T = (1, 0, 0), t = (1, 0, 1), b = (1, 1, 1), f = (0, 1, 1), F = (0, 1, 0)$, the Nmatrix below is induced, where $D_1^0= \{$$T, t, b$$\}$. Let us look at these values in more detail. By Definition~\ref{semantica_1}, an ordered triple $(x_1, x_2, x_3)$ represents the value $(\alpha, \neg\alpha, \neg{\circ}\alpha)$ -- $\circ\alpha$ is congruent to $\sim$($\alpha$ $\land$ $\neg\alpha$). Therefore, the state (1, 0, 0) is the state in which $\alpha$ holds, $\neg\alpha$ does not hold, and $\neg$$\circ\alpha$ does not hold. Thus, $\alpha$ is a positive information that is reliable in all levels. Therefore, $\alpha$ is true. The state (1, 0, 1), on the other hand, is a state in which $\alpha$ holds and $\neg\alpha$ does not hold. However, there is contradictoriness in the level of the consistency of $\alpha$, because $\circ\alpha$ holds and $\neg$$\circ\alpha$ also holds. Thus, $\alpha$ is a positive information that has conflict about its reliability. The state (1, 1, 1) is the state in which contradiction occurs in the basic level of $\alpha$ and $\neg\alpha$, therefore, $\alpha$ is a conflicting information. Similarly to the others, (0, 1, 0) is the state in which $\alpha$ is false, and (0, 1, 1) is the state in which there is negative information about $\alpha$ that has conflict about its reliability.

\begin{figure}[h]
    \centering
    \includegraphics[width=\textwidth]{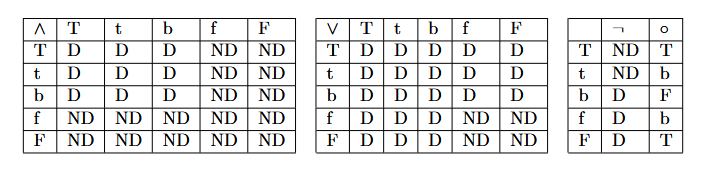}
    \includegraphics[width=.4\textwidth]{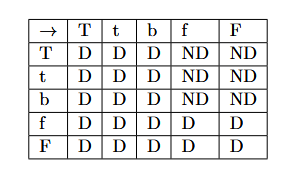}
    \caption{}
\end{figure}

\noindent\textbf{Example} Suppose that John and Joseph are arguing over the question whether $p$: ``The dog is outside'' is true. They settle this question by hearing if a dog barks outside. However, they also agree that through this method of verification, the proposition ``The dog is outside" is not reliable: we may conclude there is no dog outside when there is a dog that does not bark. We may conclude there is a dog if there is a machine outside reproducing the sound of a bark and no dog. Suppose that they hear, in fact, the barking of a dog. In that case, they both agree that there is positive information about the proposition ``The dog is outside''. However, they also agree that the consistency of this proposition is conflicting, which means that $\vartheta(p) = t$. Another question would be to ask ``Is $p$ consistent?''. To intuitively understand the answer, let us look at $p$: in one way, it is consistent at the first level of $p$ and $\neg p$. In another way, it is inconsistent at the second level of $\circ p$ and $\neg$$\circ p$. Therefore, it is expected that $\vartheta(\circ p) = b$, that is, there is conflicting information about the consistency of $p$.

\begin{theorem}{$\text{(Soundness)}$}\label{soundness}
        If $\Gamma$ $\vdash_{L_1^0}$ $\alpha$, then $\Gamma$ $\vDash_{\mathcal{M}_1^0}$ $\alpha$.
\end{theorem}

\begin{proof}
    This proof is straightforward using Nmatrices by checking, for each axiom and rule of L$_1^0$, its correspondingtruth-table, given the correspondence between swap structures and its induced Nmatrices.

\begin{figure}[h]
    \centering
    \includegraphics[width=.4\textwidth]{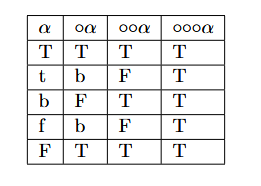}
    \includegraphics[width=.6\textwidth]{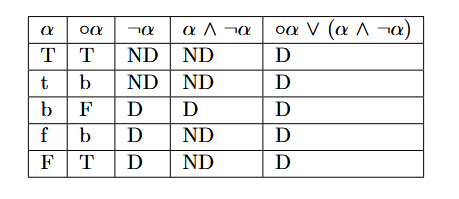}
    \caption{}
\end{figure}

    (i) $\vDash_{\mathcal{M}_1^0}$ $\circ$$\circ$$\circ$$\alpha$


(ii) $\vDash_{\mathcal{M}_1^0}$ $\circ\alpha$ $\lor$ ($\alpha$ $\land$ $\neg\alpha$)

\end{proof}

The next step is to prove completeness in relation to the swap structure. Our strategy is to prove indirectly: first, considering completeness in relation to valuations. Then, proving that both consequence relations are equivalent.

\newtheorem{lemma}{Lemma}

\begin{definition}\label{maximal non trivial}
    Let \textbf{L} be a Tarskian logic over the language $\mathcal{L}_\Sigma$, $\Gamma$ $\cup$ $\alpha$ $\subseteq$ $\mathcal{L}_\Sigma$. $\Gamma$ is maximal non-trivial with respect to $\alpha$ in $\mathcal{L}_\Sigma$ if $\Gamma$ $\nvdash$ $\alpha$ and $\Gamma$, $\beta$ $\vdash$ $\alpha$, for any $\beta$ $\notin$ $\Gamma$. \cite[p. 36]{carnielli2016}
\end{definition}

\begin{definition}\label{valuation function}
A function $\vartheta$: $\mathcal{L}_\Sigma$ $\mapsto$ \{0, 1\} is a valuation for L$_1^0$, or an L$_1^0$-valuation, if it satisfies the following clauses:
\[\begin{array}{ll}
 (vAnd)    &  \vartheta(\alpha \land \beta) = 1 \ \Leftrightarrow \  \vartheta(\alpha) = 1 \text{ and } \vartheta(\beta) = 1\\
   (vOr)  & \vartheta(\alpha \lor \beta) = 1 \ \Leftrightarrow \ \vartheta(\alpha) = 1 \text{ or } \vartheta(\beta) = 1\\
   (vImp)      &  \vartheta(\alpha \rightarrow \beta) = 1 \ \Leftrightarrow \ \vartheta(\alpha) = 0 \text{ or } \vartheta(\beta) = 1\\
  (vNeg)   & \vartheta(\neg\alpha) = 0 \ \Rightarrow \ \vartheta(\alpha) = 1\\
   (vConCiw)      &  \vartheta(\circ\alpha) = 1  \ \Leftrightarrow \ \vartheta(\alpha) = 0 \text{ or } \vartheta(\neg\alpha) = 0\\
    (vCc^1)     &  \vartheta(\neg{\circ^2}\alpha) = 1   \ \Leftrightarrow \ \vartheta(\circ^2\alpha) = 0.

\end{array}
\]
\end{definition}

Suppose $\Gamma$ $\nvdash$ $\alpha$. Then, by Theorem 2.2.6 \cite[p. 37]{carnielli2016}, for any Tarskian and finitary language over $\mathcal{L}_\Sigma$, there exists a set $\Delta$ such that $\Gamma$ $\subseteq$ $\Delta$ $\subseteq$ $\mathcal{L}_\Sigma$ and $\Delta$ is maximal non-trivial with respect to $\alpha$. 

\begin{theorem}\label{valuation2}
    Let $\Gamma$ $\cup$ $\{\alpha\}$ $\subseteq$ $\mathcal{L}_\Sigma$, $\Gamma$ being maximal non-trivial with respect to $\alpha$ in L$_1^0$. The mapping $\vartheta$: $\mathcal{L}_\Sigma$ $\mapsto$ \{0, 1\} defined by:
\[\vartheta(\beta) = 1 \Leftrightarrow \beta \in \Gamma, \text{ for all } \beta \in \mathcal{L}_\Sigma
\]

    is a valuation for L$_1^0$.
\end{theorem}

\begin{proof}
For proofs of clauses (vAnd), (vOr), (vImp), (vNeg), (vConCiw), the reader is referred to \cite[p. 38]{carnielli2016} and \cite[Theorem~3.1.3]{carnielli2016}.\\
(vCc$^1$) Suppose, by contradiction, that $\vartheta$($\circ$$\circ\beta$) = 1 and $\vartheta$($\neg{\circ}$$\circ\beta$) = 1. Then $\circ\circ$$\beta$ $\in$ $\Gamma$ and $\neg$$\circ\circ$$\beta$ $\in$ $\Gamma$, that is, $\Gamma$ $\vdash$ $\circ\circ$$\beta$ and $\Gamma$ $\vdash$ $\neg{\circ}$$\circ\beta$. By axiom (cc$^1$), (bc1) and MP, $\vdash$ $\circ\circ$$\beta$ $\rightarrow$ ($\neg{\circ}$$\circ\beta$ $\rightarrow$ $\alpha$), therefore, by two applications of MP, $\Gamma$ $\vdash$ $\alpha$, which is absurd given that $\Gamma$ is maximal non-trivial with respect to $\alpha$. Therefore, $\vartheta({\circ^2}\beta) = 1$ implies that $\vartheta(\neg{\circ^2}\beta) = 0$. Since the converse holds by (vNeg), we proved that $\vartheta$ satisfies condition (vCc$^1$).
    
\end{proof}

\begin{theorem}\label{completeness-val}
    Completeness of L$_1^0$ w.r.t. valuations: For every $\Gamma$ $\cup$ $\alpha$ $\subseteq$ $\mathcal{L}_\Sigma$:
\[
 \Gamma \vDash_{L_1^0} \alpha \Rightarrow \Gamma \vdash_{L_1^0} \alpha
\]

\end{theorem}

\begin{proof}
    By contraposition, suppose that $\Gamma$ $\nvdash$ $\alpha$. Then, there is a maximal non-trivial set $\Delta$ w.r.t. $\alpha$ such that $\Gamma$ $\subseteq$ $\Delta$. By Theorem \ref{valuation2}, there exists a valuation such that $\vartheta$($\beta$) = 1 $\Leftrightarrow$ $\beta$ $\in$ $\Delta$. Therefore, there exists a valuation such that v[$\Gamma$] $\subseteq$ \{1\}, given that $\Gamma$ $\subseteq$ $\Delta$ but $\vartheta$($\alpha$) = 0, as $\alpha$ $\notin$ $\Delta$. Therefore, $\Gamma$ $\nvDash$ $\alpha$. 
\end{proof}

    Let us prove now that any valuation $\vartheta$ for L$_1^0$ induces a valuation $h$ over  $\mathcal{M}_1^0$ which satisfies the same formulas as  $\vartheta$.  This  will be sufficient to prove completeness of L$_1^0$ w.r.t. its Nmatrix semantics.

    \begin{theorem}\label{representation theorem}
    Given a valuation $\vartheta$ for L$_1^0$, consider the function $h:\mathcal{L}_\Sigma \to \mathds{B}_1^0$  given as follows: $h(\alpha)=(\vartheta(\alpha),\vartheta(\neg\alpha),\vartheta(\neg{\circ}\alpha))$, for every $\alpha$ $\in$ $\mathcal{L}_\Sigma$. Then,  $h$ is a valuation over  $\mathcal{M}_1^0$ such that, for every $\alpha$ $\in$ $\mathcal{L}_\Sigma$, $\vartheta$($\alpha$) = 1 iff $h$($\alpha$) $\in D_1^0$.
\end{theorem}

\begin{proof} Let $\alpha$ $\in$ $\mathcal{L}_\Sigma$. By definition of $h$, $h(\alpha)=(\vartheta(\alpha),\vartheta(\neg\alpha),\vartheta(\neg{\circ}\alpha))$ and $h(\neg\alpha)=(\vartheta(\neg\alpha),\vartheta(\neg\neg\alpha),\vartheta(\neg{\circ}\neg\alpha))$. By Definition~\ref{semantica_1}, $h(\neg\alpha) \in \neg h(\alpha)$. The proof that $h(\alpha \# \beta) \in h(\alpha) \# h(\beta)$ is analogous. Now, note that  $h(\circ\alpha)=(\vartheta(\circ\alpha),\vartheta(\neg{\circ}\alpha),\vartheta(\neg{\circ}{\circ}\alpha))$ while ${\circ}h(\alpha)=\{({\sim}(\vartheta(\alpha) \land \vartheta(\neg\alpha)),\vartheta(\neg{\circ}\alpha),\vartheta(\neg{\circ}\alpha) \land {\sim}(\vartheta(\alpha) \land \vartheta(\neg\alpha)))\}$. As a consequence of Definition~\ref{valuation function}, $\vartheta({\sim}\beta)={\sim}\vartheta(\beta)$, for every $\beta$. Hence, by (vConCiw) and the other clauses, $\vartheta({\circ}\alpha) = {\sim}(\vartheta(\alpha) \land \vartheta(\neg\alpha))$. Now, by $(vCc^1)$, $\vartheta(\neg{\circ^2}\alpha) = {\sim}\vartheta(\circ^2\alpha) = {\sim}{\sim}(\vartheta({\circ}\alpha) \land \vartheta(\neg{\circ}\alpha)) = \vartheta({\circ}\alpha) \land \vartheta(\neg{\circ}\alpha) = {\sim}(\vartheta(\alpha) \land \vartheta(\neg\alpha)) \land \vartheta(\neg{\circ}\alpha)$. This proves that  $h({\circ}\alpha) \in {\circ}h(\alpha)$ and so $h$ is a valuation  over  $\mathcal{M}_1^0$. By the very definition, $\vartheta$($\alpha$) = 1 iff $h$($\alpha$) $\in D_1^0$.
\end{proof}

\begin{theorem}{$\text{(Completeness)}$}\label{completeness}
        If $\Gamma$ $\vDash_{\mathcal{M}_1^0}$ $\alpha$, then $\Gamma$ $\vdash_{L_1^0}$ $\alpha$.
\end{theorem}

\begin{proof} 
Suppose that $\Gamma$ $\vDash_{\mathcal{M}_1^0}$ $\alpha$, and let $\vartheta$ be a valuation for L$_1^0$ such that $\vartheta(\beta)=1$ for every $\beta \in \Gamma$. Let $h$ be the valuation  over  $\mathcal{M}_1^0$ induced by $\vartheta$ as in Theorem~\ref{representation theorem}. Then, $h(\beta) \in D_1^0$ for every $\beta \in \Gamma$ and so, by hypothesis, $h(\alpha) \in D_1^0$. This means that $\vartheta(\alpha)=1$, proving that $\Gamma \vDash_{L_1^0} \alpha $. By Theorem~\ref{completeness-val},  $\Gamma \vdash_{L_1^0} \alpha$.
\end{proof}
         
L$_1^k$ are logics in which $\circ$$\circ$$\circ$$\alpha$ holds. This means that conflict of information is allowed at most at one iteration of the $\circ$ operator: in the level of the reliability of the reliability of a statement, $\vartheta$($\circ$$\circ$$\alpha$) $\neq$ $\vartheta$($\neg$$\circ$$\circ$$\alpha$). In the example above, this means that we will not question the criterion to determine whether the criterion through which we judge the proposition ``The dog is barking outside'' is a reliable one or not. $\circ\circ$$\alpha$ is determined in terms of $\circ\alpha$ and $\neg$$\circ\alpha$, and the value of $\neg{\circ}$$\circ$$\alpha$ will be its Boolean complement. We give in Figure~3 a table up to $\circ$$\circ$$\circ$$\alpha$ to illustrate these notions.

\begin{figure}[h]
    \centering
    \includegraphics[width=.6\textwidth]{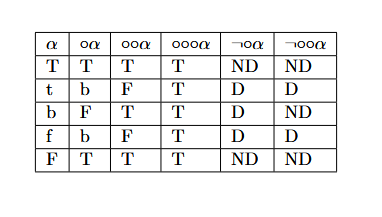}
    \caption{}
\end{figure}
One desideratum of such logics is that, when considering the negation of statements that have controversial criteria, they should also have controversial criteria. ``The dog is barking outside'', whether answered positively or negatively, remains with a weak criterion of truth. That is the reason we introduce the stronger logics L$_1^1$ and L$_1^2$:

\begin{definition}
L$_1^1$ := L$_1^0$ + (cf) $\neg\neg\alpha$ $\rightarrow$ $\alpha$ + (ce) $\alpha$ $\rightarrow$ $\neg\neg\alpha$
\end{definition}

Let Cbr = mbCciw +  (cf) $\neg\neg\alpha$ $\rightarrow$ $\alpha$ + (ce) $\alpha$ $\rightarrow$ $\neg\neg\alpha$   \cite{coniglio2025logicparaconsistentbeliefrevision}.
It is worth noting that L$_1^1 =$ Cbr + (cc$^1$) $\circ$$\circ$$\circ$$\alpha$.

In order to define a swap structure for L$_1^1$, we simply restrict the swap structure for L$_1^0$.

\begin{definition}
The swap structure for L$_1^1$  is the multialgebra 
        $\mathcal{B}_1^1$ = $\langle \mathds{B}_1^0$, $\land$, $\lor$, $\rightarrow$, $\neg$, $\circ$$\rangle$ defined as $\mathcal{B}_1^0$ (recall Definition~\ref{semantica_1}) but now with the multioperator for negation defined as follows,  for every $x$ in $\mathds{B}_1^0$:

$\begin{array}{lll}
   \neg x  & = & \{z \in \mathds{B}_1^0 \ : \ z_1 = x_2 \ \mbox{ and } \ z_2=x_1\}.
\end{array}$

\noindent The non-deterministic matrix associated to $\mathcal{B}_1^1$ is $\mathcal{M}_1^1 = (\mathcal{B}_1^1, D_1^0)$, where $\vDash$$_{\mathcal{M}_1^1}$ is the corresponding consequence relation.
\end{definition}

The Nmatrix induced by a swap structure for L$_1^1$ are the same as L$_1^0$, except for negation (see Figure~4).


\begin{figure}[!h]
    \centering
    \includegraphics[width=.2\textwidth]{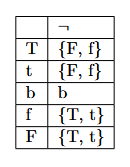}
    \caption{}
\end{figure}

\begin{theorem}
    Soundness. If $\Gamma \vdash_{L_1^1} \alpha$ then $\Gamma \vDash_{\mathcal{M}_1^1} \alpha$
\end{theorem}
\begin{proof}
    It remains to consider the additional case of axioms (ce) and (cf). Clearly, if $h$ is  a valuation over $\mathcal{M}_1^1$ then $h(\alpha) \in D_1^0$ iff $h(\neg\neg\alpha) \in D_1^0$. This can be proven by inspection of Figure~5, but also by analyzing the definition of the multioperator for $\neg$ in $\mathcal{B}_1^1$. This shows that $\mathcal{B}_1^1$ validates  axioms (ce) and (cf). 
\end{proof}

\begin{figure}[h]
    \centering
    \includegraphics[width=.6\textwidth]{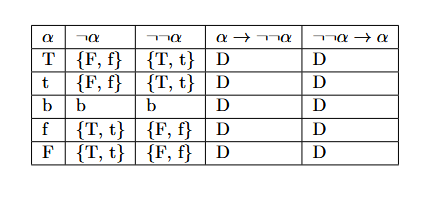}
    \caption{}
\end{figure}

To prove completeness, the same strategy used for L$_1^0$ can be used adding the following two conditions to the bivaluations. 

\begin{definition} 
A function $\vartheta: \mathcal{L}_\Sigma \mapsto \{0,1\}$ is a valuation for L$_1^1$, or an L$_1^1$-valuation, if it is an L$_1^0$ valuation and it satisfies the following clause:
\[\begin{array}{ll}
  (vdn)   & \vartheta(\neg\neg\alpha) = \vartheta(\alpha).
\end{array}
\]
\end{definition}

\begin{theorem}
     Completeness of L$_1^1$ w.r.t. valuations: For every $\Gamma$ $\cup$ $\{\alpha\}$ $\subseteq$ $\mathcal{L}_\Sigma$:

    \begin{center}
        $\Gamma$ $\vDash_{L_1^1}$ $\alpha$ $\Rightarrow$ $\Gamma$ $\vdash_{L_1^1}$ $\alpha$
    \end{center}
    
\end{theorem}
\begin{proof}
    We use the same construction of maximal non-trivials sets as in the proof of Theorem~\ref{valuation2}. In the present case, we only need to prove the validity of clause (vdn) and the completeness proof follows standardly. Thus, let $\Gamma$ be a maximal non-trivial with respect to $\alpha$ in L$_1^1$, and define $\vartheta$ as in  the proof of Theorem~\ref{valuation2}. By axioms (ce) and (cf), $\alpha \in \Gamma$ iff $\neg\neg\alpha \in \Gamma$. Hence, $\vartheta$ satisfies (vdn).
\end{proof}

The proof of the following result follows by adapting the proof for L$_1^0$:

\begin{theorem}{$\text{(Completeness)}$}
        If $\Gamma$ $\vDash_{\mathcal{M}_1^1}$ $\alpha$, then $\Gamma$ $\vdash_{L_1^1}$ $\alpha$.
\end{theorem}

In the case of L$_1^1$, it still does not meet the requirement that the controversy of the criterion must spread to its negation. To do this, it is necessary to go one level up the hierarchy.

\begin{definition}
    L$_1^2$ := L$_1^1$ + (ip$^2$) $\neg$$\circ$$\neg\alpha$ $\leftrightarrow$ $\neg$$\circ$$\alpha$
\end{definition}

\begin{definition}
The swap structure for L$_1^2$  is the multialgebra 
        $\mathcal{B}_1^2$ = $\langle \mathds{B}_1^0$, $\land$, $\lor$, $\rightarrow$, $\neg$, $\circ$$\rangle$ defined as $\mathcal{B}_1^0$ (recall Definition~\ref{semantica_1}) but now with the multioperator for negation is deterministic, and defined as follows,  for every $x$ in $\mathds{B}_1^0$:

$\begin{array}{lll}
   \neg x  & = & \{(x_2,x_1,x_3)\}.
\end{array}$

\noindent The non-deterministic matrix associated to $\mathcal{B}_1^2$ is $\mathcal{M}_1^2 = (\mathcal{B}_1^2, D_1^0)$, where $\vDash$$_{\mathcal{M}_1^2}$ is the corresponding consequence relation.
\end{definition}

\noindent
The table for negation in $\mathcal{M}_1^2$ can be seen in  Figure~6. Completeness and soundness proof for L$_1^2$ can be obtained analogously to the previous cases.

\begin{figure}[h]
    \centering
    \includegraphics[width=.15\textwidth]{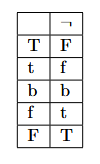}
    \caption{}
\end{figure}

\begin{theorem}{$\text{(Soundness and Completeness)}$}
        $\Gamma \vdash_{L_1^2} \alpha$ \ if and only if \  $\Gamma \vDash_{\mathcal{M}_1^2} \alpha$.
\end{theorem}

\noindent\textbf{Example} Suppose that Joseph is in his living room and he wants to know whether there is someone in his bedroom. To settle this question, he decides to use the following criterion, that he recognizes to be weak: the proposition $p$: ``There is someone in my bedroom'' is the case iff if I kick a ball inside it, someone complains about it. Suppose that Joseph kicks a ball and no one complains about it. What is the value of the proposition ``There is no one in my bedroom'', the logical negation of $p$?

Given that no one complained about Joseph's ball, it means that John will accept that $p$ is not the case, that is, $\vartheta(p) \notin$ D. Moreover, $p$, although being decided as false, is controversial with relation to its method of verification. Therefore, $\vartheta(p) = f$. Thus, $\vartheta(\neg p) = t$, which means that the proposition ``There is no one in my bedroom'' is decided in the affirmative by Joseph, however, its method of verification is not reliable, because $\vartheta(\circ\neg p) = b$.

As the examples above show, in the subfamily L$_1^k$ of logics, the consistency of the consistency of some statement can be questioned, that is, we can question whether some method that decides a given question is itself reliable. From a philosophical perspective, this relates to the problem of criterion, a classical discussion in epistemology: given a proposition, how do we determine the criteria of its truth? Suppose that, given a proposition $p$, we give a set of criteria $\Delta$ such that $\Delta$ specifies the truth conditions for $p$. In our context, $\Delta$ decides whether $p$ and $\neg$$p$ are designated or not. However, we can question whether $\Delta$ are reliable criteria. That is, we ask ``What are the conditions for $\Delta$ to be correct?''. In that case, we are asking whether $\neg$$\circ p$ is designated or not. Now, the problem is how to determine whether $\Delta$ are correct criteria. In L$_1^k$, however, the value of $\circ$$\circ$$\alpha$ and $\neg$$\circ$$\circ$$\alpha$, that is, the reliability of the $\Delta$ criteria, are not themselves subjected to controversiality. Therefore, we stop questioning the reliability of propositions in the second level. In the example above, this means that it does not make sense to ask ``Is the criteria through which I decide that kicking a ball is not a reliable criterion to determine $p$ reliable criteria themselves?''.

\section*{LFI3: an extension of L$_1^2$}

    LFI1 and LFI2 were introduced as 3-valued LFIs in \cite{carnielli2000} (i.e., logics characterized by 3-valued logical matrices). Notwithstanding, LFI1 was introduced many years before in 1970 in \cite{itala1970}. In that work, it was introduced as J3 under the signature $\Sigma_{J3}$ = \{$\neg, \lor, \triangledown$\}. The main motivation behind the conception of J3 is a problem as old as the history of paraconsistent logic itself: Ja\'skowski's challenge to find logics that are paraconsistent, intuitive and applicable \cite{jaskowski1969propositional}. J3 is one of the strong candidates to meet these requirements, because it is a 3-valued paraconsistent logic, which makes it arguably more intuitive than other non-deterministic semantics available at the time. By recovering classical reasoning, J3 can also meet the requirement of practical application of this logic, given that it is expressive enough to deal with paraconsistent and classical phenomena. 

    Nevertheless, LFI1 or J3 may be regarded as simplifying consistency \textit{too much}, due to the equivalence between consistency and non-contradiction in LFI1, that is, $\vdash_{LFI1} \circ\alpha \leftrightarrow {\sim}(\alpha \land \neg\alpha)$. A 5-valued logic, on the other hand, can distinguish more cases than a 3-valued logic. Thus, LFI3 is introduced as a more general solution to Ja\'skowski's challenge. It is possible to extend L$_1^2$ so that it becomes a 5-valued LFI (i.e., an LFI characterized by a  5-valued logical matrix).

    The key for moving from a 5-valued Nmatrix (for L$_1^2$) to a 5-valued matrix (for LFI3) is adding restrictions to the (strictly non-deterministic) multioperators for the binary connectives in order to obtain deterministic operators. A very natural way to do this, generalizing the idea of LFI1, is to consider a linear order in the domain $\mathds{B}_1^0$ as follows: $F < f < b < t < T$. This defines automatically the infimum (minimum) and the supremum (maximum) between two elements, interpreting conjunction and disjunction in a deterministic way. Moreover, the negation $\neg$ satisfies: $a \leq b$ iff $\neg b \leq \neg a$. From this, the De Morgan laws hold automatically (as in LFI1). Using the (deterministic) strong negation ${\sim} a := \neg a \land {\circ} a$, it is possible to define the (deterministic) implication as in LFI1 as follows: $a \to b:={\sim}a \vee b$. This produces the (deterministic) truth-tables displayed in Figure~7.

    \begin{figure}[h]
    \centering
    \includegraphics[width=.6\textwidth]{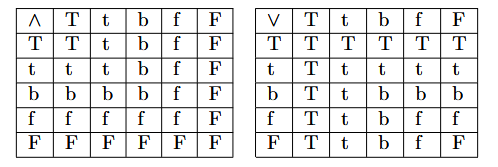}
    \includegraphics[width=.6\textwidth]{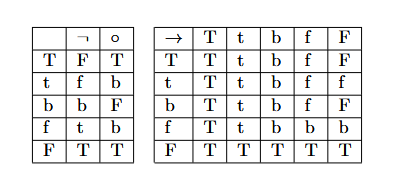}
    \caption{}
\end{figure}

Using these tables, a logical matrix can be defined, in which the set of designated values is $D_1^0$.  A natural question is how to represent  the new operators $\land$, $\vee$, $\sim$ and $\to$ in terms of snapshots. Since they are deterministic, they are generated by a twist structure instead of a swap structure. Such a  twist structure is a (deterministic) sub-multialgebra of the multialgebra of  L$_1^2$ given by its swap structure. Moreover, by the considerations above, it is enough defining the twist operator for conjunction which corresponds to the truth-table of Figure~7. Using that twist operator and the given one for negation, the other ones are obtained automatically.

Now, observe that, for $a,b \in \mathds{B}_1^0$ (seen as snapshots), its infimum $a \ \dot\land \ b$ is expressed, as a twist operator, as follows: 
$$(a_1, a_2, a_3)\ \land\ (b_1, b_2, b_3) = (a_1 \land b_1,a_2 \vee b_2,G(a_1, a_2, a_3,b_1, b_2, b_3))$$
for some Boolean function $G$. Note that the first coordinate of  $a \ \dot\land \ b$ follows from the swap structure of L$_1^2$, while the second one follows from the definition of the negation and the fact that the De Morgan law must  hold. By using  the technique,  based on Karnaugh Maps, for constructing twist structures from truth-tables presented in \cite[p. 974-75]{veronicakarnaugh}, it is possible to determine a Boolean function $G$ satisfying the requirements. This lead us to the following definition: 

\begin{definition} \label{twist:LFI3}
    The twist structure for LFI3 is the algebra $T_{\bf 2} = (\mathds{B}_1^0, \dot\land, \dot\lor, \dot\rightarrow, \dot\neg, \dot\circ)$ over $\Sigma$ such that the operations are defined as follows, for every $(a_1, a_2, a_3), (b_1, b_2, b_3)$ in $\mathds{B}_1^0$: 
    \[\begin{array}{ll}
        1. & (a_1, a_2, a_3)\ \dot\land\ (b_1, b_2, b_3) = (a_1 \land b_1, a_2 \lor b_2, ({\sim} a_2 \land b_3) \lor (a_3\ \land {\sim} b_2) \lor (a_3 \land b_3)) \\
        2. & (a_1, a_2, a_3)\ \dot\lor\ (b_1, b_2, b_3) = (a_1 \lor b_1, a_2 \land b_2, ({\sim} a_1 \land b_3) \lor (a_3\ \land {\sim} b_1) \lor (a_3 \land b_3))\\
        3.  & (a_1, a_2, a_3)\ \dot\rightarrow\ (b_1, b_2, b_3) = (a_1 \rightarrow b_1, b_2 \land ({\sim} a_2 \lor a_3),\\ 
        &  \hspace*{1cm} ({\sim} a_2 \land b_3) \lor ({\sim} a_2 \land a_3\ \land {\sim} b_1) \lor (a_3 \land b_3) \lor ({\sim} a_1 \land a_3\ \land {\sim} b_1))\\
        4. & \dot\neg\ (a_1, a_2, a_3) = (a_2, a_1, a_3)\\
        5. & \dot\circ\ (a_1, a_2, a_3) = ({\sim}(a_1 \land a_2),\ a_3,\ a_3\ \land {\sim}(a_1 \land a_2))
    \end{array}\]
    The matrix associated to $T_{\bf 2}$ is $\mathcal{M}_{\bf 2} = \langle T_{\bf 2},D_1^0 \rangle$.
\end{definition}


It is clear that the matrix $\mathcal{M}_{\bf 2}$ coincides with the original one for LFI3 displayed in Figure~7. Let $\sim$$\alpha$ denote a strong negation, defined by $\neg\alpha \land \circ\alpha$. The corresponding truth-table is displayed in Figure~8.

\begin{figure}[h]
    \centering
    \includegraphics[width=.15\textwidth]{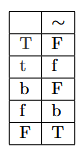}
    \caption{}
\end{figure}

It is easy to prove that the negation $\sim$ recovers explosion, and also satisfies excluded middle.

\begin{theorem}
    For all formulas $\alpha$, $\beta$, 
    \begin{enumerate}
        \item $\alpha$, $\sim$$\alpha$ $\vDash_{LFI3}$ $\beta$
        \item $\vDash_{LFI3}$ $\alpha$ $\lor$ $\sim$$\alpha$
    \end{enumerate}
\end{theorem}

The following simple result (see~\cite[Lemma~2.3]{coniglio2019}) will be useful:

\begin{lemma} \label{IncMat}
Let $L_1=\langle {\bf A}_1, D_1\rangle$  and $L_2=\langle {\bf A}_2, D_2\rangle$ be two matrix logics defined over a signature $\Theta$ such that
${\bf A}_2$ is a subalgebra of ${\bf A}_1$ and $D_2 = D_1 \cap A_2$. Then $L_1$ is a sublogic of $L_2$, that is: for every $\Gamma \cup \{\psi\}$, if $\Gamma \vdash_{L_1} \psi$ then $\Gamma \vdash_{L_2} \psi$.
\end{lemma}

\begin{theorem}
    LFI3 is a sublogic of classical logic. 
\end{theorem}
\begin{proof}
    The two-element Boolean algebra $\{T, F\}$ is a subalgebra of the algebra $T_{\bf 2}$ of LFI3, and $T$ is the only designated value of LFI3 in $\{T, F\}$. The result follows by Lemma~\ref{IncMat}.
\end{proof}

\begin{theorem}\label{sublogiclfi1}
    LFI3 a sublogic of LFI1.
\end{theorem}

\begin{proof} By identifying $1$, $1/2$ and $0$ with $T$, $b$ and $F$, respectively, we see that  the algebra $A_{LFI1} = \langle\{1, 1/2, 0\}, \land, \to, \neg, \circ\rangle$ of LFI1 is a subalgebra of the algebra $T_{\bf 2}$ of LFI3. Moreover, the designated values of LFI1 are the ones of LFI3 which belong to the domain of $A_{LFI1}$. The result follows by Lemma~\ref{IncMat}.
\end{proof}

\begin{definition}{(Definition 2.1 \cite[p. 126]{coniglio2019})}
    Let $L_1$ and $L_2$ be two distinct standard propositional logics defined over the signature $\Sigma$ such that $L_1$ is a proper sublogic of $L_2$. $L_1$ is maximal w.r.t. $L_2$ if, for every formula $\alpha$ over $\Sigma$: if $\vdash_{L_2} \alpha$ but $\nvdash_{L_1} \alpha$ then $L_1^+$ obtained from $L_1$ by adding $\alpha$ as a (schema) theorem coincides with $L_2$.  
\end{definition}

\begin{lemma}{(Theorem 2.4 \cite[p. 128]{coniglio2019})}\label{maximalitylemma}
    Let $L_1 = \langle\bm{A_1}, D_1\rangle$ and $L_2 = \langle\bm{A_2}, D_2\rangle$ be two distinct finite matrix logics over a same signature $\Sigma$ s.t. $\bm{A_2}$ is a subalgebra of $\bm{A_1}$ and $D_2 = D_1 \cap A_2$. Assume the following:
    \begin{enumerate}
        \item $A_1 = \{0, 1, a_1, ..., a_k, a_{k+1}, ..., a_n\}$ and $A_2 = \{0, 1, a_1, ..., a_k\}$ are finite such that $0 \notin D_1$, $1 \in D_2$ and $\{0,1\}$ is a subalgebra of $\bm{A_2}$.
        \item There are formulas $\top(p)$ and $\bot(p)$ in $\mathcal{L}_\Sigma$ depending at most on one variable $p$ such that $\vartheta(\top(p))=1$ and $\vartheta(\bot(p)) = 0$, for every valuation $\vartheta$  for $L_1$.
        \item For every $k+1 \leq i \leq n$ and $1 \leq j \leq n$ (with $i \neq j$), there exists a formula $\alpha_j^i(p)$ in $\mathcal{L}_\Sigma$ depending at most on one variable $p$ such that, for every valuation $\vartheta$ for $L_1$, $\vartheta(\alpha_j^i(p))=a_j$ if $\vartheta(p) = a_i$.
    \end{enumerate}
    Then, $L_1$ is maximal w.r.t $L_2$.
\end{lemma}

\begin{theorem}\label{maximalitylfi3}
    LFI3 is maximal w.r.t. LFI1.
\end{theorem}

\begin{proof}
From Theorem \ref{sublogiclfi1}, condition (1) from Lemma \ref{maximalitylemma} is satisfied straightforwardly. Defining $\top(p):= \circ^3p$ and $\bot(p) := \neg\top(p)$, condition (2) is also satisfied. For item (3), we need to prove that for the values $\{t,f\}$ of LFI3, there are truth functions that output all truth values in $\{T, t, b, f, F\}$. The truth tables in Figure \ref{fig:tabela-lfi3} show that the existence of formulas $\alpha_j^t(p)$ and $\alpha_j^f(p)$ such that:  $\vartheta(\alpha_j^t(p))=j$ and $\vartheta(\alpha_j^f(p)) = j$, for all $j \in \{T, t, b, f, F\}$, if $\vartheta(p) = t$ or $f$, respectively. Therefore, by Lemma \ref{maximalitylemma}, LFI3 is maximal w.r.t. LFI1.
\end{proof}

\begin{figure}[htbp]
    \centering
    \includegraphics[width=.4\linewidth]{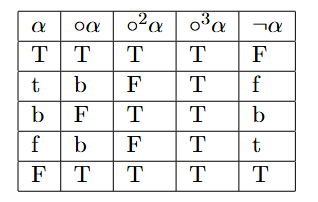}
    \caption{}
    \label{fig:tabela-lfi3}
\end{figure}

\begin{corollary}\label{corollary}
    LFI1 = LFI3 + (cc) $\circ\circ$$\alpha$
\end{corollary}

\begin{proof}
    This follows straightforwardly from Theorem \ref{maximalitylfi3} and the fact that $\vdash_{LFI1} \circ^2\alpha$ and $\nvdash_{LFI3} \circ^2\alpha$. 
\end{proof}

Corollary \ref{corollary} presents a new axiomatization to LFI1. 

Let $\alpha \leftrightarrow \beta$ be defined as $(\alpha \rightarrow \beta) \land (\beta \rightarrow \alpha)$. It is possible to define a congruence relation in LFI3.

\begin{definition}
    $\alpha \equiv \beta$ := $(\alpha \leftrightarrow \beta) \land (\neg\alpha \leftrightarrow \neg\beta) \land (\neg$$\circ\alpha \leftrightarrow \neg$$\circ\beta$).

\end{definition} 

\begin{figure}[h]
    \centering
    \includegraphics[width=.4\textwidth]{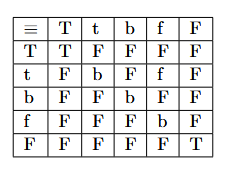}
    \caption{}
\end{figure}

\begin{theorem}
    $\equiv$ is a congruence relation.
\end{theorem}

\begin{proof}
    It is straightforward to check that $\equiv$ is a congruence given the fact that $h(\alpha \equiv \beta) \in D_1^0$ iff $h(\alpha) = h(\beta)$, for every valuation $h$ over $T_{\bf 2}$ (see Figure~10).
\end{proof}

Using the technique for obtaining a Hilbert calculus from a twist structure semantics presented in \cite[p. 976-77]{veronicakarnaugh}, it is possible to determine the axioms for LFI3 based on the twist structure given in Definition~\ref{twist:LFI3}.

\begin{definition} [Hilbert calculus for LFI3]
  The logic LFI3, presented as a Hilbert calculus, is defined as follows: LFI3 := L$_1^2$ plus the following axioms:
\[
\begin{array}{ll}
  A1  &  \neg(\alpha \land \beta) \leftrightarrow (\neg\alpha \lor \neg\beta)\\
   A2  & \neg(\alpha \lor \beta) \leftrightarrow (\neg\alpha \land \neg\beta)\\
    A3   &  \neg(\alpha \rightarrow \beta) \leftrightarrow (\neg\beta \land (\sim\neg\alpha \lor \neg{\circ}\alpha))\\
   A4  & \neg{\circ}(\alpha \land \beta) \leftrightarrow ((({\sim}\neg\alpha \land \neg{\circ}\beta) \lor (\neg{\circ}\alpha \land {\sim}\neg\beta)) \lor (\neg{\circ}\alpha \land \neg{\circ}\beta))\\
     A5  &  \neg{\circ}(\alpha \lor \beta) \leftrightarrow ((({\sim}\alpha \land \neg{\circ}\beta) \lor (\neg{\circ}\alpha \land {\sim}\beta)) \lor (\neg{\circ}\alpha \land \neg{\circ}\beta))\\
    A6 & \neg{\circ}(\alpha \rightarrow \beta) \leftrightarrow\\
    & \hspace*{2mm}   (({\sim}\neg\alpha \land \neg{\circ}\beta) \lor ({\sim}\neg\alpha \land \neg{\circ}\alpha \land {\sim}\beta) \lor (\neg{\circ}\alpha \land \neg{\circ}\beta) \lor ({\sim}\alpha \land \neg{\circ}\alpha \land {\sim}\beta))\\
\end{array}
\] 
If there is a derivation in LFI3 of $\alpha$ from $\Gamma$ we will write $\Gamma \vdash_{LFI3} \alpha$.
\end{definition}

\begin{theorem} [Soundness and completeness of the Hilbert calculus of LFI3] \ \\
 $\Gamma \vdash_{LFI3} \alpha$ \ if and only if \  $\Gamma \vDash_{LFI3} \alpha$.
\end{theorem}

\noindent\textbf{Remark}
As we have seen along this paper, the defined connective $\bullet\alpha:= \neg{\circ}\alpha$ played an important role as a component of the snapshots of the swap and twist structures considered here. In the context of LFIs, such a connective is known as an {\em inconsistency operator}. Next result analyzes the connection  of inconsistency with contradiction in LFI3.

\begin{theorem} The following properties hold in LFI3:
\[
\begin{aligned}
\begin{array}{ll}
1. & \alpha \land \neg\alpha \vdash \circ^2\alpha \\
2. & \circ^2\alpha \nvdash \alpha \land \neg\alpha\\
3. & \circ\alpha \nvdash \circ^2\alpha
\end{array}
\qquad
\begin{array}{ll}
4. & \circ^2\alpha \nvdash \circ\alpha\\
5. & \bullet\alpha \nvdash \alpha \land \neg\alpha\\
6. & \alpha \land \neg\alpha \vdash \bullet\alpha
\end{array}
\end{aligned}
\]
\end{theorem}

\begin{proof}
    Using soundness and completeness, we can straightforwardly verify the properties above with 5-valued truth tables for LFI3  (see Figure~11).

\begin{figure}[h]
    \centering
    \includegraphics[width=.6\textwidth]{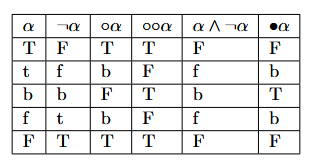}
    \caption{}
\end{figure}
\end{proof}

Another interesting feature of LFI3 is that it hides other recovery operators. 

\begin{definition}
    $\circ^*\alpha := (\alpha \land \circ\alpha \land \circ\circ\alpha) \lor (\neg\alpha \land \circ\alpha \land \circ\circ\alpha)$

    \begin{figure}[h]
    \centering
    \includegraphics[width=.2\textwidth]{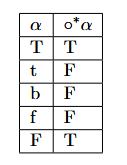}
    \caption{}
\end{figure}
\end{definition}

\begin{theorem}$\circ^*\alpha$ is a consistency operator, that is,
\[\begin{array}{ll}
   1. &  \circ^*\alpha, \alpha, \neg\alpha \vdash \beta\\
   2. & \circ^*\alpha, \alpha \nvdash \beta\\
   3. & \circ^*\alpha, \neg\alpha \nvdash \beta\\
\end{array}\]

\begin{proof}
    The proof is straightforward using truth tables (see Figure~12).
\end{proof}

\end{theorem}

Furthermore, the $\circ^*$ operator also validates (cc$^*$) $\circ$$^*\circ$$^*\alpha$. This operator recovers the idea of consistency as preservation of classical values. Indeed, with this operator it is possible to recover classical logic inside LFI3.

\begin{definition}
    $\circ^{\#}\alpha$ := $\circ^*\alpha$ $\lor$ ($\bullet\alpha \land \circ$$\circ\alpha$) (see Figure~13).

    \begin{figure}[h]
    \centering
    \includegraphics[width=.2\textwidth]{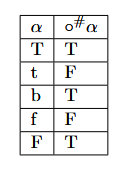}
    \caption{}
\end{figure}
\end{definition}

Analogously to the $\circ^*$ operator, $\circ^\#$ recovers LFI1 reasoning. Indeed, with $\circ^\#$ we can recover LFI1 inside LFI3 as well. Here, we will consider a general semantical proof of recovery. From now on, $Var(\Gamma)$ will denote the set of propositional variables occurring in a set of formulas $\Gamma$.


\begin{theorem}\label{geralzao}
Let $L_{1}=\langle {\bf A}_{1}, D_1\rangle$  and $L_2=\langle {\bf A}_2, D_2\rangle$ be two matrix logics in the conditions of Lemma~\ref{IncMat}. Suppose that, in addition there  exist a formula $\circ'(p)$ depending on a single propositional variable $p$ such that, for every valuation $\vartheta$ over $L_1$, $\vartheta(\circ'\alpha) \in D_1$ iff $\vartheta(\alpha) \in {\bf A}_2$. Then:  $\circ'(Var(\Gamma \cup \{\alpha\}), \Gamma \vdash_{L_{1}} \alpha$ iff $\Gamma \vdash_{L_2} \alpha$, where $\circ'(X) = \{\circ'(p) \ : \ p \in X\}$, for every set $X$ of propositional variables.
\end{theorem}
\begin{proof}
Suppose that ${\bf A}_{1} =\{a_1, \ldots, a_n, \ldots, a_l\}$ and ${\bf A}_{2} =\{a_1, \ldots, a_n\}$, where $n<l$. Let us consider a formula 
$\circ_{L_n}(p)$ satisfying the required hypothesis.  Then, the recovery operator $\circ_{L_n}$ can be represented as the  truth function displayed in Figure~14, where $\mathcal{T} \in D_1$ and $\mathcal{F} \in {\bf A}_{1}-D_1$ are arbitrary.
\\($\Rightarrow$) Suppose that $\circ_{L_n}(Var(\Gamma \cup \{\alpha\}), \Gamma \vdash_{L_{1}} \alpha$. Let $\vartheta$ be a valuation over $L_2$ such that $\vartheta[\Gamma] \subseteq D_2$. Since ${\bf A}_2$ is a subalgebra of ${\bf A}_1$ then $\vartheta$ can be seen as a valuation over $L_1$ such that  $\vartheta(p) \in  {\bf A}_2$ for every $p \in Var(\Gamma \cup \{\alpha\})$. This means that $\vartheta[\circ_{L_n}(Var(\Gamma \cup \{\alpha\})] \subseteq D_1$ and $\vartheta[\Gamma] \subseteq D_1$, since $D_2 \subseteq D_1$.  By hypothesis, $\vartheta(\alpha) \in D_1$. From this,  $\vartheta(\alpha) \in D_2$,  since $D_2=D_1 \cap {\bf A}_2$. This proves that $\Gamma \vdash_{L_2} \alpha$.
\\$(\Leftarrow)$ Suppose that $\Gamma \vdash_{L_2} \alpha$. Let $\vartheta$ be a valuation over $L_1$ such that $\vartheta[\circ_{L_n}(Var(\Gamma \cup \{\alpha\})] \subseteq D_1$ and $\vartheta[\Gamma] \subseteq D_1$. By the properties of $\circ_{L_n}(p)$ if follows that $\vartheta(p) \in  {\bf A}_2$ for every $p \in Var(\Gamma \cup \{\alpha\})$. Let $\bar\vartheta$ be a valuation over  $L_2$ such that $\bar\vartheta(p)=\vartheta(p)$ for every   $p \in Var(\Gamma \cup \{\alpha\})$. Then, $\bar\vartheta(\beta)=\vartheta(\beta)$ for every formula $\beta$ such that $Var(\beta) \subseteq Var(\Gamma \cup \{\alpha\})$, since ${\bf A}_2$ is a subalgebra of ${\bf A}_1$. But then $\bar\vartheta$ is a valuation over  $L_2$ such that $\bar\vartheta[\Gamma] \subseteq D_2$, since $D_2=D_1 \cap {\bf A}_2$. By hypothesis, $\bar\vartheta(\alpha) \in D_2$. This means that $\vartheta(\alpha) \in D_1$, proving that  $\circ'(Var(\Gamma \cup \{\alpha\}), \Gamma \vdash_{L_{1}} \alpha$.
\end{proof}

   \begin{figure}[h]
    \centering
    \includegraphics[width=.2\textwidth]{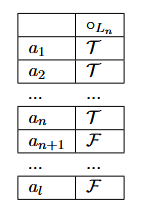}
    \caption{}
\end{figure}

\begin{theorem}
For all sets of formulas $\Gamma$ and formula $\alpha$,
\[
\circ^\#(Var(\Gamma \cup \{\alpha\}), \Gamma \vdash_{LFI3} \alpha \ \text{ iff } \ \Gamma \vdash_{LFI1} \alpha
\]
\end{theorem}
\begin{theorem} For all sets of formulas $\Gamma$ and formula $\alpha$,
\[
\circ^*(Var(\Gamma \cup \{\alpha\}), \Gamma \vdash_{LFI3} \alpha \ \text{ iff } \ \Gamma \vdash_{CPL} \alpha
\]
\end{theorem}

\begin{proof}
    Both cases are particular cases of Theorem \ref{geralzao}.
\end{proof}
    

\section*{mbCciw: a logic for skeptics}

As we have seen in the previous section, the subfamily L$_1^k$ of logics represents a philosophical - or pragmatic - viewpoint in which questionability makes sense only up to the level of the reliability of the reliability of a sentence. From there on, we assume to know whether the criteria are correct or incorrect.

Recall from Definition~\ref{mbcciw} that mbCciw := mbC + (ciw). In mbCciw, we face a different situation: there is no iteration of $\circ^n$ operators such that $\circ^n$$\alpha$ is a theorem. In terms of Nmatrices, the $\circ$ operator behaves as shown in Figure~15 \cite[p. 263]{carnielli2016}.

\begin{figure}[h]
    \centering
    \includegraphics[width=.2\textwidth]{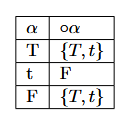}
    \caption{}
\end{figure}

The value $T$ is interpreted as true, $t$ as contradictory, and $F$ as false. Notice that if $\alpha$ is true or false, we can always question its reliability up to any level. This goes in line with the skeptical view defended by Sextus Empiricus, in his chapter \textit{Does There Exist a Criterion of Truth?}, Book II of \textit{Outlines of Pyrrhonism}:

\begin{quote}
    This dispute, then, will either declare to be decidable or to be undecidable; if undecidable, they will be granting at once that judgment should be suspended; but if decidable, let them say with what it is to be decided, seeing that we do not have any agreed-upon criterion and do not know - indeed, are inquiring - whether one exists. And anyhow, in order to decide the dispute that has arisen about the criterion, we have need of an agreed-upon criterion it is necessary first to have decided the dispute about the criterion. Thus, with reasoning falling into the circularity mode, finding a criterion becomes aporetic; for we do not allow them to adopt a criterion hypothetically, and if they wish to decide about the criterion by means of a criterion we force them into an infinite regress. \cite[p. 129]{sextus} 
\end{quote}

That is, if we try to settle one question through some set of criteria, we can always question the reliability of such criteria, because mbCciw allows us to do so. Furthermore, we can question the same about the criteria of criteria, and so on, up to any finite level. Therefore, skepticism finds its adequate logical multiplicity in mbCciw. 

From now on, it will be useful to analyze the swap structures characterization of mbCciw and some of its axiomatic extensions. 

Snapshots for mbCciw  are pairs $x=(x_1,x_2)$ representing 0-1 values for $\alpha$ and $\neg\alpha$. By (TND) $x_1 \vee x_2=1$. This produces a domain formed by 3 snapshots:  $\mathds{B}_0 = \{x \in {\bf 2}\times {\bf 2} \ : \ x_1 \lor x_2 = 1\} = \{(1,0),(1,1),(0,1)\}=\{$T,t,F$\}$ under the identifications T$=(1,0)$ (true), t$=(1,1)$ (contradictory), and F$=(0,1)$ (false).

\begin{definition}\label{swap-mbCciw}
   The swap structure for  mbCciw is the multialgebra 
        $\mathcal{B}_0$ = $\langle \mathds{B}_0$, $\land$, $\lor$, $\rightarrow$, $\neg$, $\circ$$\rangle$  over $\Sigma$ such that the multioperations are defined as in $\mathcal{B}_1^0$ (recall Definition~\ref{semantica_1}), but now over the domain  $\mathds{B}_0$. Let $D_0 = \{x \in \mathds{B}_0 \ : \ x_1 = 1\} = \{$T,t$\}$ be the set of designated values. The Nmatrix associated to $\mathcal{B}_0$ is $\mathcal{M}_0 = (\mathcal{B}_0, D_0)$.
\end{definition}

The 3-valued Nmatrix $\mathcal{M}_0$ is the usual Nmatrix for mbCciw (see~\cite[p. 263]{carnielli2016}). It is displayed in Figure~16.

\begin{figure}[h]
    \centering
    \includegraphics[width=\textwidth]{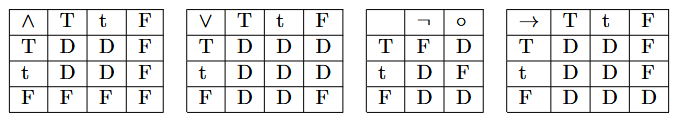}
    \caption{}
\end{figure}

\section*{L$_0^k$: logics for dogmatists}

\begin{definition}
    L$_0^0$ = mbCciw + (cc) $\circ\circ$$\alpha$ = mbCci
\end{definition}

In the case of mbCci, we consider 3-valued Nmatrices that validate (cc), as shown by the table in Figure~17. Controversiality only appears at the level of a proposition and its negation: we cannot even question the reliability of the reliability of a proposition. That is, given a proposition $\alpha$, our criteria are either reliable and produce a decision, or they are unreliable and controversial. In comparison with skepticism, L$_0^0$ contains the logical multiplicity of a dogmatic position, albeit allowing for contradictory states: given a certain set of criteria, we determine which propositions are true and which are not. However, we do not question such criteria one level above: we assume we have the criteria at hand. Soundness and completeness for L$_0^0$ are presented in \cite[p. 67 and p. 266]{carnielli2016}.

\begin{definition}
The swap structure $\mathcal{B}_0^0$ for  L$_0^0$ is obtained from $\mathcal{B}_0$,  the one for mbCciw, by making the multioperator for $\circ$  deterministic:

$\begin{array}{lll}
   \circ x  & = & \{({\sim}(x_1 \land x_2),x_1 \land x_2)\}.
\end{array}$
\end{definition}

The operator $\circ$ for  L$_0^0$ can be seen in Figure~17.

\begin{figure}[h]
    \centering
    \includegraphics[width=.25\textwidth]{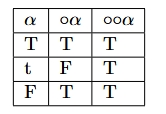}
    \caption{}
\end{figure}

\begin{definition}
    L$_0^1$ = L$_0^0$ + (dn) $\neg\neg\alpha$ $\leftrightarrow$ $\alpha$
\end{definition}

The logic L$_0^1$ is also another known logic, considered in the context of paraconsistent belief revision as Cie \cite[p. 4]{coniglio2025logicparaconsistentbeliefrevision}.

\begin{definition}
The swap structure $\mathcal{B}_0^1$ for L$_0^1$ is obtained from $\mathcal{B}_0^0$,  the one for  mbCci, by changing the multioperator for $\neg$ to a deterministic one:

$\begin{array}{lll}
 \neg x = \{(x_2, x_1)\}.
\end{array}$
\end{definition}

The operator $\neg$ for  L$_0^1$ can be seen in Figure~18.



\begin{figure}[h]
    \centering
    \includegraphics[width=.2\textwidth]{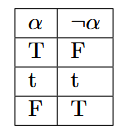}
    \caption{}
\end{figure}

\noindent\textbf{Example} Suppose that Judas is a priest who wants to know whether God exists or not. He then proposes to himself the following: if God exists, then a sufficient sign for this is that the roses in his garden shall flourish. If God does not exist, then a sufficient sign for this is that a scarecrow appears in his kitchen tomorrow. Surprisingly, the other day his entire garden flourishes and a scarecrow appears on top of his fridge. What can he conclude from this? 

In this example, the priest dogmatically accepts - that is, he does not question the reliability of the reliability - the criteria for the truth of the proposition $p$: ``God exists'' and its negation. However, these criteria determine that both \textit{p} and $\neg p$ hold, that is, $\vartheta(p) = t$. Therefore, his criteria - not the criteria of the criteria - are not reliable because they produce conflict of decision, i.e., $\vartheta(\circ p) = F$. 

Imagine, on the other hand, that only Judas' garden had flourished - i.e., no scarecrow has appeared in his house. In that case, he concludes that God exists, because he dogmatically accepts his criterion as leading to truth (there is no logical multiplicity to question his criteria). 

\section*{Pragmatics and Controlled Consistency}

In everyday life, criteria are more flexible depending on the context. Sometimes, people are required not only to justify how they arrived at a certain proposition, but also how they know this justification is correct. However, the chain of reasons must stop at some point: in everyday life, skepticism is not a practical threat, unless you take the pragmatic implications of skepticism to its full development. The hierarchy of logics L$_k^n$ offers an adequate logical multiplicity that is able to deal with any level of control \textit{k} that is required given a certain context. 

Let us think about a general way to give a semantics for this hierarchy. Characterizing these logics in terms of Nmatrices has the advantage of making these logics more intuitive. However, as we go up in the hierarchy, complexity increases. In L$_2^n$, there are 8 truth values; in L$_3^n$, there are 13 truth values. Therefore, this method becomes progressively less treatable and less intuitive.

One positive result from this characterization: we can answer one of the old problems for LFIs, namely, to present LFIs that are more than 3-valued. We have shown here how to get to a 5-valued LFI, but the same method (although more complex) will lead to the deterministic extensions of an 8-valued LFI, a 13-valued LFI, etc. A better way to tackle this problem in general is to think in terms of restricted Nmatrices (RNmatrices) \cite{coniglio2020simple}.  From now on, $\mathbf{Val}(\mathcal{M})$ will denote the set of valuations over a given Nmatrix $\mathcal{M}$.

\begin{definition} \label{def:RNmatrix}
A {\em restricted Nmatrix} (RNmatrix) based on an Nmatrix $\mathcal{M} = \langle \mathcal{A},D \rangle$ is a pair $\mathcal{R}=\langle \mathcal{M},\mathcal{F}\rangle$ such that $\mathcal{F} \subseteq \mathbf{Val}(\mathcal{M})$. The consequence relation associated to $\mathcal{R}$ is $\vDash_{\mathcal{R}}$, such that $\Gamma \vDash_{\mathcal{R}} \alpha$ iff, for every $\vartheta \in \mathcal{F}$, if $\vartheta[\Gamma] \subseteq D$ then $\vartheta(\alpha) \in D$.
\end{definition}

Recall  from Definition~\ref{swap-mbCciw} the swap structure for mbCciw, which gives origin to a 3-valued Nmatrix $\mathcal{M}_0$ with domain $\mathcal{B}_0=\{ T,t,F \}$ (see Figure~16). The set of designated values is $D_0=\{ T, t\}$.

Now, given the logic L$_n^0$, we can validate (cc$^n$) $\circ^{n+2}$$\alpha$ in $\mathcal{M}_0$ in the following way.  Consider the set $\mathcal{F}_n^0 \subseteq \mathbf{Val}(\mathcal{M}_0)$ such that $\vartheta \in \mathcal{F}_n^0$ iff it  satisfies the following condition:  if $\vartheta(\circ^n\alpha) \in \{T, F\}$, then $\vartheta(\circ^{n+1}\alpha) =$ T. Let $\mathcal{R}_n^0=\langle \mathcal{M}_0,\mathcal{F}_n^0\rangle$ be the corresponding RNmatrix.

\begin{theorem}
    $\vDash_{\mathcal{R}_n^0} \circ^{n+2}\alpha$.
\end{theorem}
\begin{proof}
Let $\vartheta \in \mathcal{F}_n^0$ and $\alpha$ a formula. If $\vartheta(\circ^n\alpha) \in \{$T, F$\}$  then $\vartheta(\circ^{n+1}\alpha) =$ T, by definition of $\mathcal{F}_n^0$. If, on the contrary, $\vartheta(\circ^n\alpha)=$ t then $\vartheta(\circ^{n+1}\alpha) =$ F, by definition of $\circ$ in $\mathcal{M}_0$. Hence, we conclude that  $\vartheta(\circ^{n+1}\alpha) \in \{$T, F$\}$, for every  $\vartheta \in \mathcal{F}_n^0$ and formula $\alpha$. But then, by definition of $\mathcal{F}_n^0$, $\vartheta(\circ^{n+2}\alpha) =$ T.
\end{proof}

 The situation for  $n \geq 3$ is shown in Figure~19. Observe that, if $n=0$ then the condition reduces to the following: if $\vartheta(\alpha) \in \{T, F\}$, then $\vartheta(\circ\alpha) =$ T. In this case, the second column of the table in Figure~19 contains only the values T and F, while the third contains only T. This validates the formula $\circ^2\alpha$, and recovers the deterministic table for $\circ$ in L$_0^0$ shown in Figure~17.

\begin{figure}[h]
    \centering
    \includegraphics[width=\textwidth]{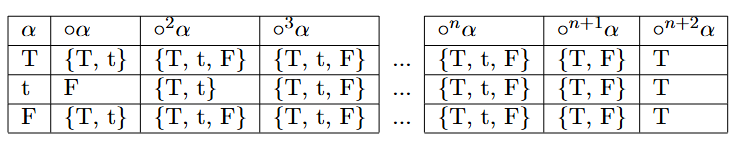}
    \caption{}
\end{figure}

Consider the concrete example when $n=2$:

\begin{definition}
L$_2^0$ := mbCciw + (cc$^2$) $\circ^4\alpha$ 
\end{definition}

Let $\mathcal{F}_2^0 \subseteq \mathbf{Val}(\mathcal{M}_0)$ be such that $\vartheta \in \mathcal{F}_2^0$ iff it  satisfies the following condition:  if $\vartheta(\circ^2\alpha) \in \{T, F\}$, then $\vartheta(\circ^{3}\alpha) =$ T. Let $\mathcal{R}_2^0=\langle \mathcal{M}_0,\mathcal{F}_2^0\rangle$ be the corresponding RNmatrix. As we can see from the non-deterministic RNmatrix in Figure~20, (cc$^2$) is sound in $\mathcal{R}_2^0$. In such table, T$_*$ indicates that the value T was mandatory, according to the definition of $\mathcal{F}_2^0$.

\begin{figure}[h]
    \centering
    \includegraphics[width=.4\textwidth]{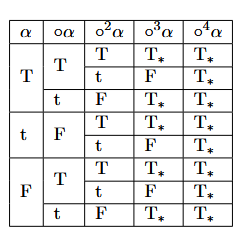}
    \caption{}
\end{figure}



    

The completeness proof for the case of L$_n^0$ is analogous to the construction of maximal non-trivial sets presented for $L_1^0$ and $L_1^1$.

\begin{theorem}\label{completenessL2}
    Let $\Delta \cup \{\alpha\} \subseteq \mathcal{L}_\Sigma$ such that $\Delta$ is maximal non-trivial with respect to $\alpha$ in L$_n^0$. Consider the mapping $\vartheta: \mathcal{L}_\Sigma \mapsto \{T, t, F\}$ defined by
\[
\vartheta(\beta)=
\begin{cases}
T, & \text{if } \beta \in \Delta, \neg\beta \notin \Delta,\\
t,  & \text{if } \beta \in \Delta, \neg\beta \in \Delta,\\
F, & \text{if } \beta \notin \Delta, \neg\beta \in \Delta
\end{cases}
\]
Then, $\vartheta$ is a valuation in $\mathcal{F}_n^0$.
\end{theorem}
\begin{proof}
The fact that $\vartheta \in \mathbf{Val}(\mathcal{M}_0)$ is well-know in the literature of LFIs, taking into account that L$_n^0$ is an axiomatic extension of mbCciw. Hence, it only remains to prove that $\vartheta$ satisfies  the RNmatrix restriction: if $\vartheta(\circ^n\alpha) \in \{$T, F$\}$, then $\vartheta(\circ^{n+1}\alpha) = $T.

Thus, assume that  $\vartheta(\circ^n\alpha) \in \{$T, F$\}$. By definition of $\vartheta$, it follows that either $\circ^n\alpha \not\in \Delta$ or  $\neg{\circ}^n\alpha \not\in \Delta$, and so $\circ^n\alpha \land \neg{\circ}^n\alpha \not\in \Delta$. Therefore, ${\sim}(\circ^n\alpha \land \neg{\circ}^n\alpha) \in \Delta$. From this,  $\circ^{n+1}\alpha \in \Delta$, since $\circ\beta$ is equivalent to ${\sim}(\beta \land \neg\beta)$ in L$_n^0$, for every formula $\beta$. Since $\Delta$ is a closed theory then $\circ^{n+2}\alpha \in \Delta$, because it is a theorem in L$_n^0$. But then, by the considerations above, $\circ^{n+1}\alpha \land \neg{\circ}^{n+1}\alpha \not\in \Delta$. Given that $\circ^{n+1}\alpha \in \Delta$, we conclude that $\neg{\circ}^{n+1}\alpha \not\in \Delta$. From this, $\vartheta(\circ^{n+1}\alpha) = $T.
\end{proof}

Strengthening negation in terms of RNmatrices is also possible: 
\begin{definition}
    L$_2^1$ := L$_2^0$ + (dn) $\alpha$ $\leftrightarrow$ $\neg\neg\alpha$
\end{definition}

In this case, we can use the Nmatrix for Cbr \cite{coniglio2025logicparaconsistentbeliefrevision} as a basis. In general:

\begin{definition}{$(L_n^1\ \text{ Logics})$} Let  $\mathcal{M}_1$ be the Nmatrix for Cbr with domain $\mathcal{B}=\{ T,t,F \}$ and set of designated values $D=\{ T, t\}$  (see Figure~21). The RNmatrix for L$_n^1$ is $\mathcal{R}_n^1=\langle \mathcal{M}_1,\mathcal{F}_n^1\rangle$ such that  $\mathcal{F}_n^1 \subseteq \mathbf{Val}(\mathcal{M}_1)$ is defined as follows: $\vartheta \in \mathcal{F}_n^1$ iff it  satisfies 
\[(vCc^n) \text{ if } \vartheta(\circ^n\alpha) \in \{T, F\}, \text{ then } \vartheta(\circ^{n+1}\alpha) = T.\]
\end{definition}

\begin{figure}[h]
    \centering
    \includegraphics[width=\textwidth]{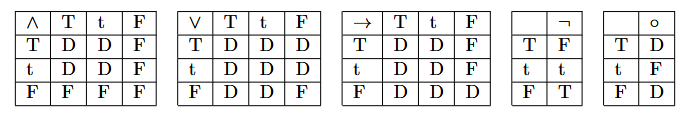}
    \caption{}
\end{figure}

    The argument for soundness of ($cc^n$) here is analogous to the case of L$_n^0$ logics.
    
\begin{definition}{$(L_n^2\ Logics)$}
Let $\mathcal{M}_1$  be the non-deterministic matrix for Cbr. The RNmatrix for L$_n^2$ is $\mathcal{R}_n^2=\langle \mathcal{M}_1,\mathcal{F}_n^2\rangle$ such that  $\mathcal{F}_n^2 \subseteq \mathbf{Val}(\mathcal{M}_1)$ is defined as follows:  $\vartheta \in \mathcal{F}_n^1$ iff it  satisfies  condition $(vCc^n)$ from L$_n^1$ and the condition below.
    $$(vip^1) \ \ \vartheta(\circ\alpha) = \vartheta({\circ}\neg\alpha).$$ 
\end{definition}

\begin{theorem}
$\vDash_{\mathcal{R}_n^2} \neg{\circ}\neg\alpha \leftrightarrow \neg{\circ}\alpha$
\end{theorem}

\begin{proof} The non-deterministic truth-table in Figure~22 shows the validity of (ip$^2$).
\end{proof}

\begin{figure}[h]
    \centering
    \includegraphics[width=.7\textwidth]{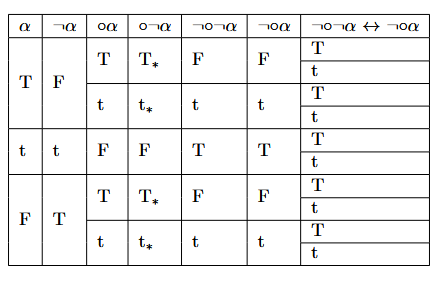}
    \caption{}
\end{figure}

\begin{definition}{(L$_n^{k}$ Logics)}
Let $k \geq 2$,  and let $\mathcal{M}_1$  be the non-deterministic matrix for Cbr. The RNmatrix for L$_n^k$ is $\mathcal{R}_n^k=\langle \mathcal{M}_1,\mathcal{F}_n^k\rangle$ such that  $\mathcal{F}_n^k \subseteq \mathbf{Val}(\mathcal{M}_1)$ is defined as follows:  $\vartheta \in \mathcal{F}_n^k$ iff it  satisfies  condition $(vCc^n)$ from L$_n^1$ and the  conditions below.
  \[\begin{array}{ll}
       (vip^1)  & \vartheta(\circ\alpha) = \vartheta({\circ}\neg\alpha) \\
       (vip^2) & \vartheta(\circ^2\alpha) = \vartheta(\circ^2\neg\alpha)\\
       \dots  &\\
       (vip^{k-1}) & \vartheta(\circ^{k-1}\alpha) = \vartheta(\circ^{k-1}\neg\alpha)\\
   \end{array}\]
    
\end{definition}

\begin{theorem}{(Soundness of L$_n^{k}$)}
  \ $\vdash_{L_n^{k}} \alpha$ implies $\vDash_{\mathcal{R}_n^{k}} \alpha$.
\end{theorem}
\begin{proof}
    Let us prove in general that L$_n^1$ plus condition (vip$^{j}$) for the RNmatrix validates the axiom ($ip^j$) $\neg$$\circ^{j}\alpha \leftrightarrow \neg$$\circ^{j}\neg$$\alpha$. The argument is pretty straightforward: given that $\vartheta(\circ^{j}\alpha) = \vartheta(\circ^{j}\neg\alpha)$, then $\vartheta(\neg$$\circ^{j}\alpha) = \vartheta(\neg$$\circ^{j}\neg\alpha)$, since the negation in $\mathcal{M}_1$  is deterministic. Therefore, $\vartheta(\neg$$\circ^{j}\alpha \leftrightarrow \neg$$\circ^{j}\neg\alpha) \in$ D, because of the tables in $\mathcal{M}_1$.
\end{proof}

In order to prove completeness, the following result will be fundamental.

\begin{theorem}\label{theorem_savior} Let $1 \leq j \leq k$. Then,
\ $\vdash_{L_n^{k}} \circ^j\alpha \leftrightarrow \circ^j\neg\alpha$.
\end{theorem}
\begin{proof}
    Let $j=1$. Since $\vdash_{L_n^{k}} \neg\neg\alpha \leftrightarrow \alpha$ then $\vdash_{L_n^{k}} (\neg\alpha \land \neg\neg\alpha) \leftrightarrow  (\alpha \land \neg\alpha)$ and so $\vdash_{L_n^{k}} {\sim}(\neg\alpha \land \neg\neg\alpha) \leftrightarrow  {\sim}(\alpha \land \neg\alpha)$. From this, $\vdash_{L_n^{k}} \circ\alpha \leftrightarrow {\circ}\neg\alpha$. Now, assume that $\vdash_{L_n^{k}} \circ^j\alpha \leftrightarrow \circ^j\neg\alpha$ for $1 \leq j \leq k-1$ (induction hypothesis). Note that $\vdash_{L_n^{k}} \neg{\circ}^j\alpha \leftrightarrow \neg{\circ}^j\neg\alpha$, by axiom (ip$^j$). Then $\vdash_{L_n^{k}} (\circ^j\alpha \land \neg{\circ}^j\alpha) \leftrightarrow  (\circ^j\neg\alpha \land \neg{\circ}^j\neg\alpha)$ and so $\vdash_{L_n^{k}} {\sim}(\circ^j\alpha \land \neg{\circ}^j\alpha) \leftrightarrow  {\sim}(\circ^j\neg\alpha \land \neg{\circ}^j\neg\alpha)$. From this, $\vdash_{L_n^{k}} \circ^{j+1}\alpha \leftrightarrow \circ^{j+1}\neg\alpha$. By reasoning inductively, we obtain the desired result.
\end{proof}

\begin{theorem}{$(Completeness\ of\ L_n^{k})$}
\ $\vDash_{\mathcal{R}_n^{k}} \alpha$ implies $\vdash_{L_n^{k}} \alpha$
\end{theorem}

\begin{proof}
In the present case, we will only consider condition $(vip^j)$ in a general way, and we will prove that the valuation induced by the maximal non-trivial set is in fact a valuation that satisfies $(vCc^n)$ and $(vip^j).$ Let $\vartheta$ be a mapping defined similarly to the case of L$_2^0$ in Theorem \ref{completenessL2}.

\noindent
$(vCc^n)$ Suppose that $\vartheta(\circ^{n}\alpha) \in \{$T,F$\}$. Then, either $\circ^{n}\alpha \not\in \Delta$ or  $\neg{\circ}^{n}\alpha \not\in \Delta$. From this, $\circ^{n}\alpha \land \neg{\circ}^{n}\alpha \not\in \Delta$ and so ${\sim}(\circ^{n}\alpha \land \neg{\circ}^{n}\alpha) \in \Delta$. This implies that $\circ^{n+1}\alpha \in \Delta$. Suppose, by contradiction, that  $\neg{\circ}^{n+1}\alpha \in \Delta$. By reasoning as above, it follows that $\circ^{n+2}\alpha \notin \Delta$. This contradicts the fact that $\circ^{n+2}\alpha \in \Delta$, by axiom (cc$^n$) and the fact that $\Delta$ is a closed theory. Therefore  $\neg{\circ}^{n+1}\alpha \notin \Delta$. From this,  $\vartheta(\circ^{n+1}\alpha) =$ T.\\[1mm]

$(vip^{j})$ Let $1 \leq j \leq k-1$. By contradiction, suppose that $\vartheta(\circ^j\alpha) \neq \vartheta(\circ^j\neg\alpha)$. By Theorem~\ref{theorem_savior}, there are two cases to consider:

\begin{itemize}
    \item[] [$\vartheta(\circ^j\alpha) = $T and $\vartheta(\circ^j\neg\alpha)=$ t] In that case, $\vartheta(\neg{\circ}^j\alpha) =$  F and $\vartheta(\neg{\circ}^j\neg\alpha)=$ t. That is, $\neg{\circ}^j\neg\alpha \in \Delta$ and $\neg{\circ}^j\alpha \not\in \Delta$. From axiom $(ip^j)$, $\vdash_{L_n^{k}} \neg{\circ}^j\alpha \leftrightarrow \neg{\circ}^j\neg\alpha$. Therefore, $(\neg{\circ}^j\alpha \leftrightarrow \neg{\circ}^j\neg\alpha) \in \Delta$ and so $\neg{\circ}^j\alpha \in \Delta$, a contradiction.

    \item[] [$\vartheta(\circ^j\alpha)=$ t and $\vartheta(\circ^j\neg\alpha) =$ T] The reasoning is analogous to the previous case. 
\end{itemize}

We don't need to consider cases in which $\circ^j\alpha$ and $\circ ^j\neg\alpha$ are not both designated, because of Theorem \ref{theorem_savior}.
\end{proof}

Now, it will be proven that L$_n^{n+1}$ is a fixed point in the hierarchy, in the following sense: L$_n^{n+1} =$ L$_n^{m}$, for every $m \geq n+2$.

The key for this result is given in the following lemma:

\begin{lemma}
Let $n \geq 0$ and $m \geq n+ 2$. Then, for every $\alpha$ and $\beta$, $\neg{\circ}^m \alpha \vdash_{L_n^{n+1}}  \beta$.
\end{lemma}
\begin{proof} Suppose, by contradiction,  that $\neg{\circ}^m \alpha \nvdash_{L_n^{n+1}}  \beta$ for some $n \geq 0$, $m \geq n+2$ and formulas  $\alpha$ and $\beta$.
Then, there exists a set $\Delta$ such that  $\neg{\circ}^m \alpha \in \Delta$, and $\Delta$ is  maximal non-trivial  with respect to $\beta$ in L$_n^{n+1}$. By axiom  $(cc^n)$ we have that  ${\circ}^m \alpha={\circ}^{n+2}{\circ}^{m-(n+2)}\alpha$ is also in  $\Delta$, and so  ${\circ}^m\alpha \land \neg{\circ}^m \alpha \in \Delta$. From this, ${\circ}^{m+1} \alpha \notin \Delta$. But  ${\circ}^{m+1} \alpha={\circ}^{n+2}{\circ}^{m-(n+1)}\alpha$ is in $\Delta$, by $(cc^n)$, a contradiction.
\end{proof}

\begin{corollary} \label{coro:fixpoint}
Let $n \geq 0$ and $m \geq n+ 2$. Then, $\vdash_{L_n^{n+1}} \neg$$\circ$$^{m}$$\neg\alpha$ $\leftrightarrow$ $\neg{\circ}^{m}$$\alpha$.
\end{corollary}

\begin{theorem}{(Fixed Point Theorem)}
Let $n \geq 0$ and $m \geq n+ 2$. Then, L$_n^{n+1} =$ L$_n^{m}$.
\end{theorem}

\begin{proof} 
We need to prove that (ip$^j$) $\neg$$\circ$$^{j}$$\neg\alpha$ $\leftrightarrow$ $\neg{\circ}^{j}$$\alpha$ holds in L$_n^{n+1}$ for $n+1 \leq j < m$. The case $n+2 \leq j < m$ follows from Corollary~\ref{coro:fixpoint}. Consider now the case $j=n+1$. Since $\circ\beta$ is equivalent to ${\sim}(\beta \land \neg \beta)$ in L$_n^{n+1}$ then, for every $\Delta$ maximal non-trivial  with respect to a formula $\gamma$ in  L$_n^{n+1}$, it holds that
\[(\ast) \ \ \ \circ\beta \in \Delta \ \mbox{ if and only if } \ \beta \not\in \Delta \ \mbox{ or } \ \neg\beta \not\in \Delta.\]
Suppose, by contradiction,  that $\neg{\circ}^{n+1} \alpha \nvdash_{L_n^{n+1}}  \neg{\circ}^{n+1} \neg\alpha$.
Then, there exists a set $\Delta$ such that  $\neg{\circ}^{n+1} \alpha \in \Delta$, and $\Delta$ is  maximal non-trivial  with respect to $\neg{\circ}^{n+1} \neg\alpha$ in L$_n^{n+1}$. Since $\neg{\circ}^{n+1} \neg\alpha \notin \Delta$ then ${\circ}^{n+1} \neg\alpha \in \Delta$ and so, by Theorem \ref{theorem_savior}, ${\circ}^{n+1} \alpha \in \Delta$. By~$(\ast)$, $\circ^{n+2}\alpha \notin \Delta$, which is a contradiction because of axiom   $(cc^n)$. From this,  $\neg{\circ}^{n+1} \alpha \vdash_{L_n^{n+1}}  \neg{\circ}^{n+1} \neg\alpha$ and so  $\vdash_{L_n^{n+1}}  \neg{\circ}^{n+1} \alpha  \to \neg{\circ}^{n+1} \neg\alpha$, since L$_n^{n+1}$ satisfies the deduction metatheorem. The proof of the converse implication is analogous. This shows that axiom (ip$^{n+1}$) is also derivable in  L$_n^{n+1}$. Since L$_n^{m}$ is the extension of  L$_n^{n+1}$ obtained by adding axioms  (ip$^j$) for $n+1 \leq j < m$, we conclude that both logics coincide.
\end{proof}

With this general presentation of the $L_n^k$ family of LCCs, we have successfully presented a tool for LFIs that allows for a refined control of consistency at any level of iteration $n$ of the $\circ$ operator. A unified picture for this hierarchy with RNmatrices makes them ready to use and opens new directions in the study of LFIs. 
The philosophical idea of controlling paraconsistent reasoning according to the context brings back to LFIs ideas of da Costa, whereas its hierarchy unifies many different LFIs such as mbCci, Cbr, Cie, and introducing many new LFIs at the same time. The case-study of the basic logics of controlled consistency shows how these different degrees of paraconsistent commitment may be applied in real-life examples of reasoning. Furthermore, the development of LFI3 represents a significant advance for paraconsistent logic in light of its historical connections with LFI1 and many-valued paraconsistent logics. 

\section*{Funding}
Marcelo Coniglio acknowledges support from the National Council for Scientific and Technological
Development (CNPq), research grant 309830/2023-0. Rafael Ongaratto was supported by the S\~ao Paulo Research Foundation (FAPESP), S\~ao Paulo, Brazil [grant numbers 2024/00807-6, 2025/01892-0]. Both authors were also supported by the S\~ao Paulo Research Foundation (FAPESP), through the thematic project {\em Rationality, logic and probability -- RatioLog}, grant number 2020/16353-3.


\end{document}